\documentclass[3p]{elsarticle}
\usepackage[utf8]{inputenc}

\IfFileExists{mtpro2.sty}{%
  \usepackage[subscriptcorrection,mtphrb]{mtpro2}
  \usepackage{times}
}{\usepackage{stix}}

\usepackage{textcomp}
\usepackage{mathrsfs}
\usepackage{mathtools,amsthm,amsfonts}
\usepackage{microtype} 
\usepackage{caption}

\newtheorem{theorem}{\textbf{Theorem}}[section]
\newtheorem{lemma}[theorem]{\textbf{Lemma}}
\theoremstyle{remark}
\newtheorem{remark}{\textit{Remark}}[section]

\usepackage[bf,compact,small,pagestyles]{titlesec}
\usepackage{titletoc}
\usepackage{appendix}

\newcommand{\RR}{\mathbb R}

\newcommand{\defeq}{\stackrel{\text{def}}{=\joinrel=}}

\renewcommand{\i}{\textrm{i}}
\newcommand{\e}{\textrm e}
\renewcommand{\Re}{\operatorname{Re}}
\renewcommand{\Im}{\operatorname{Im}}

\let\bs\boldsymbol

\newcommand{\myol}[2][3]{{}\mkern#1mu\overline{\mkern-#1mu#2}}

\newcommand{\lp}{\left(}
\newcommand{\rp}{\right)}
\newcommand{\lb}{\left[}
\newcommand{\rb}{\right]}

\newcommand{\der}[3]{\ifcat#31{
								\ifnum#3=1 {d#1\over d#2}
								\else{d^#3#1\over d#2^#3}
								\fi}
				\else{d^#3#1\over d#2^#3}
				\fi}
\newcommand{\pder}[3]{\ifcat#31{
								\ifnum#3=1 {\partial#1\over \partial#2}
								\else{\partial^#3#1\over \partial#2^#3}
								\fi}
				\else{\partial^#3#1\over \partial#2^#3}
				\fi}

\newcommand{\nablaT}{\nabla_\text{T}}
\newcommand{\DeltaT}{\Delta_\text{T}}
\newcommand{\nT}{\bs n_\text{T}}

\newcommand{\tlength}{L}

\begin{document}
\begin{frontmatter}
\title{Acoustic boundary layers as boundary conditions}

\author{Martin Berggren\corref{coa}}
\ead{martin.berggren@cs.umu.se}
\author{Anders Bernland\corref{}}
\ead{anders.bernland@cs.umu.se}
\author{Daniel Noreland\corref{}}
\ead{daniel.noreland@cs.umu.se}
\address{Department of Computing Science, Umeå University, Sweden}
\cortext[coa]{Corresponding author}   

\begin{abstract}
The linearized, compressible Navier--Stokes equations can be used to model acoustic wave propagation in the presence of viscous and thermal boundary layers.
However, acoustic boundary layers are notorious for invoking prohibitively high resolution requirements on numerical solutions of the equations. 
We derive and present a strategy for how viscous and thermal boundary-layer effects can be represented as a boundary condition on the standard Helmholtz equation for the acoustic pressure.
This boundary condition constitutes an $O(\delta)$ perturbation, where $\delta$ is the boundary-layer thickness, of the vanishing Neumann condition for the acoustic pressure associated with a lossless sound-hard wall.
The approximate model is valid when the wavelength and the minimum radius of curvature of the wall is much larger than the boundary layer thickness. 
In the special case of sound propagation in a cylindrical duct, the model collapses to the classical Kirchhoff solution.
We assess the model in the case of sound propagation through a compression driver, a kind of transducer that is commonly used to feed horn loudspeakers.
Due to the presence of shallow chambers and thin slits in the device, it is crucial to include modeling of visco--thermal losses in the acoustic analysis.
The transmitted power spectrum through the device calculated numerically using our model agrees well with computations using a hybrid model, where the full linearized, compressible Navier--Stokes equations are solved in the narrow regions of the device and the inviscid Helmholtz equations elsewhere.
However, our model needs about two orders of magnitude less memory and computational time than the more complete model.

\end{abstract}


\end{frontmatter}

\section{Background}

The classical wave equation and its time harmonic counterpart, the Helmholtz equation, provide accurate mathematical models for
acoustic wave propagation under a wide range of conditions. 
Effects that are not accounted for in the linear regime are mainly related to various loss mechanisms, manifested in two different ways. 
On the one hand there is bulk loss, which is a consequence of bulk shear, heat conduction and molecular exchange of energy. 
Bulk losses are small and have in general an appreciable effect only for propagation over long distances~\cite[\S~6.4]{MoIn68}.

Boundary effects, on the other hand, are due to heat exchange with walls and viscous dissipation owing to the shear motion caused by the contact or non-slip boundary condition at the wall boundary. 
The relative importance of the thermal and viscous losses varies with the type of medium, the wavelength, and the
characteristic size of the domain.  
In air, thermal and viscous boundary effects are of the same order of magnitude at audio frequencies, since the Prandtl number for air is close to unity. 
For devices such as hearing aids, microphones, and micro-loudspeakers, these effects can have a great influence on the generated or detected sound power level. 
Another type of devices that often cannot be modeled with sufficient accuracy without a satisfactory damping model are musical wind instruments. Acoustic damping is also important when studying the damped vibrations of Micro Electro Mechanical Systems (MEMS) structures.
As we will see in the analysis of \S~\ref{BLT}, a main feature of systems for which visco-thermal losses need to be taken into account is that the quotient of the total solid-surface area to the air volume is significant.

All the above-mentioned damping effects may be accounted for by including the appropriate constitutive relations in the linearized, compressible Navier--Stokes equations.
There are software packages that can carry out numerical solutions of the full set of linearized Navier--Stokes equations for acoustic purposes, but the computational cost is generally very high, also for devices that are acoustically small.
The reason for the high computational cost is partly the introduction of four extra variables, three components for the acoustic velocity and one for temperature (or entropy) fluctuations, compared to the classical wave equation formulated in terms of the acoustic pressure only. 
An even more serious problem is the viscous and thermal boundary layers, which typically are very thin in relation to the free-space wavelength and the characteristic dimensions of the geometry.
Thus, the computational mesh has to be extremely fine in the vicinity of solid boundaries to resolve the large gradients in the boundary layers.
A recommendation from one of the major software providers~\cite[Ch.~5]{Co17AM} is to hybridize and use the full Navier--Stokes equations only when absolutely necessary, such as in thin slits, and couple these equations to the pressure Helmholtz equation for the rest of the domain. 

There is a long history of efforts to approximately account for boundary losses, going back at least to Kirchhoff~\cite{Ki68}.
One approach has been to consider particular geometries for which exact or approximate solutions to the linearized Navier--Stokes equations can be established. 
For waveguides, the solutions is typically represented in a 1D analysis by a complex propagation constant $\phi$, so that the acoustic pressure at axial position $z$ satisfies $p(z,t) = (A\e^{-\phi z} + B\e^{\phi z})\e^{\i\omega t}$, where $A$ and $B$ are amplitudes of waves propagating in the positive and negative $z$-direction, respectively.
For the case of propagation in a circular pipe, Tijdeman~\cite{Ti75} reviews and summarizes a large number of results in terms of four nondimensional parameters.
With a similar 1D-analysis, Richards~\cite{Ri86} covers the case of waves propagating between two infinite parallel plates and also, approximately, certain non-cylindrical tube geometries.
The case of arbitrary cross sections requires numerical solutions in general~\cite{Cu93}.

Another approach, also with a long history, is to use boundary-layer analysis.
Two recent expositions of the technique are by Rienstra \& Hirschberg~\cite[\S~4.5]{RiHi15} and Searby et al.~\cite{SeNiHaLa08}.
A typical use of the boundary-layer analysis has been to calculate the propagation constant for waveguides in the limit of a large radii~\cite[\S~4.5.3]{RiHi15}.
As opposed to the approach reviewed by Tijdeman~\cite{Ti75}, the results of the boundary-layer analysis will be independent of the cross-section shape, but the results are not valid for small radii in the order of the boundary-layer thickness.
Another use of the boundary-layer analysis is suggested by Searby et al.~\cite{SeNiHaLa08},  who propose a two-step procedure for calculation of the total visco--thermal losses in cavities.
The first step consists of numerically solving the Helmholtz equation for the acoustic pressure in the whole cavity.
With this  isentropically calculated pressure at the boundary as input data, the total visco--thermal boundary losses are then computed in a second step  using boundary-layer theory.
Bossart et al.~\cite{BoJoBr03} takes this idea one step further in a predictor--corrector manner and propose to recompute the outer problem using the boundary-layer solution in order to modify the wavenumber and an admittance boundary condition.

The aim of this article is to derive and propose a boundary condition that can be supplied to isentropic acoustics models, such as the Helmholtz equation for the acoustic pressure, in order to account for visco--thermal boundary losses in numerical simulations. 
The derivation is based on a boundary-layer analysis of the linearized, compressible Navier--Stokes equations.
The basic idea is to rewrite the equation of mass conservation in the boundary layer as an equation with constant acoustic density and pressure, which would be the case under isentropic conditions, and compensate for the error, to first order, by a modified wall-normal boundary condition for the acoustic velocity.
For the pressure Helmholtz equation, the final boundary condition turns out to be a so-called Wentzell (or Venttsel') condition, a  generalization of an impedance (or Robin) boundary condition, including a surface Laplacian of the pressure. 
The starting point of the derivation is a flat-wall boundary-layer analysis, and it therefore assumes that the boundary-layer thickness is small compared to the isentropic wavelength as well as the maximal radius of curvature of the wall. 
Based on the asymptotic analysis by Schmidt et al.~\cite{ScThJo14}, and using a different approach than ours, the viscous--but not the thermal--part of the proposed boundary conditions have previously been derived in a quite recent report by  Schmidt and Thöns--Zueva~\cite{ScTh14}.

The difficulty with boundary layers whose resolution demands a very fine mesh, thus requiring significant computational resources, are present also in many turbulent flow simulations. 
Common approaches to reduce the computational burden here include modelling the velocity within the boundary layer using so called wall functions~\cite{LaSp74}, or setting a modified Dirichlet~\cite{BaMiCaHu07} or Robin~\cite{JoUt15} boundary condition for the velocity a short distance away from the wall. 
Although with a similar motivation as in the present study, the form of the boundary conditions, and their applications to turbulence models outside the boundary layer, are entirely different.

We show that a typical problem setup for cavity acoustics using the proposed boundary conditions leads to a well-posed mathematical problem.
Moreover, we show that our model yields the same expression for the total visco--thermal power losses in a cavity as discussed by Searby et al.~\cite{SeNiHaLa08} and, in the special case of propagation in wave guides, that our model yields the same dispersion relation as derived by Rienstra \& Hirschberg~\cite[\S~4.5]{RiHi15}.
The limits of applicability of our model is quantitatively assessed in the case of propagation in narrow wave guides with a circular cross section, a case in which exact solution of the linearized Navier--Stokes equations are available.
The proposed boundary-condition is straightforward to implement in an existing finite-element software, which we demonstrate by implementing it using Comsol Multiphysics' so-called Weak Form PDE Interface~\cite[Ch.~16]{Co17Ref}.
Our model is tested on the case of wave propagation through a compression driver, a type of transducer used to feed mid-range horns, and we show that visco--thermal losses are important to account for in this application and that the proposed model yields essentially the same results on the transmitted power spectrum as a simulation using a hybrid Navier--Stokes/Helmholtz approach at a fraction of the computation cost in terms of CPU time and computer memory usage.

\section{Visco--thermal acoustic equations}

Since our focus is on  acoustics  in air under atmospheric conditions, the starting point for the modeling is the compressible Navier--Stokes equations together with  standard constitutive relations.
That is, air is regarded as an ideal gas and a Newtonian fluid with constant viscosity and satisfying the Stokes hypothesis, and the heat flux satisfies Fourier’s law with a constant thermal conductivity.

Linearizing  the compressible Navier--Stokes equations around quiescent air at static pressure, density, and temperature $p_0$, $\rho_0$, and $T_0$, we obtain in frequency domain the system
\begin{subequations}\label{linNS}
\begin{align}
\i\omega\rho + \rho_0\nabla\cdot\bs U &= 0,				\label{linMassC}
\\
\i\omega\bs U  + \frac1{\rho_0}\nabla p - \nu\big(\Delta\bs U + \frac13\nabla(\nabla\cdot\bs U)\big) 
&= \bs 0,								\label{linMomC}
\\
\i\omega\rho_0 c_V T + p_0\nabla\cdot\bs U - \kappa\Delta T &= 0,
\label{linEnerC}
\end{align}
\end{subequations}
where $\rho$, $p$, $\bs U$, and $T$ are the complex amplitudes of the acoustic disturbances in density, pressure, velocity, and temperature, $\omega$ the angular frequency, $\nu$ the kinematic viscosity coefficient, $c_V$ the  specific heat capacity at constant volume, and $\kappa$ the thermal conductivity.
We use the phase convention $\e^{\i\omega t}$, so that in time domain, the acoustic pressure, for instance, will be $P(\bs x, t) = \Re p(\bs x)\e^{\i\omega t}$, and we assume that system~\eqref{linNS}  is driven through an inhomogeneous boundary condition or through a wave originating in the far field.

The static conditions satisfy the ideal gas law
\begin{equation}\label{idealgas}
p_0 = r \rho_0 T_0,
\end{equation}
where the gas constant $r$ is given by the difference of the specific heats at constant pressure and volume,
\begin{equation}\label{r}
r = c_p - c_V,
\end{equation}
Linearizing the ideal gas law, we find that the acoustic disturbances satisfy
\begin{equation}\label{linIGL}
\frac p{p_0} = \frac\rho{\rho_0} + \frac T{T_0}.
\end{equation}
Moreover, the speed of sound $c$ satisfies
\begin{equation}\label{c2}
c^2 = \gamma\frac{p_0}{\rho_0},
\end{equation}
where 
\begin{equation}\label{gamma}
\gamma=\frac{c_p}{c_V}.
\end{equation}

Since the viscosity 
and thermal conductivity coefficients are very small, visco--thermal effects may in most parts of the domain be ignored, which leads to isentropic conditions, and thus that the speed of sound in these particular regions (but not generally) will satisfy
\begin{equation}\label{c2ie}
c^2 = \frac{p}{\rho}.
\end{equation}  
As a consequence, under isentropic assumptions, the variables $p$, $\rho$, and $T$ will, by equation~\eqref{linIGL}, be proportional to each other, and the system~\eqref{linMassC}--\eqref{linMomC} will then reduce to the following wave equation in first-order, frequency-domain form,
\begin{subequations}\label{wave1ord}
\begin{align}
\frac{\i\omega }{c^2}p + \nabla\cdot\rho_0\bs U &= 0,	
\label{wave1ord_mc}	\\
\i\omega\rho_0 \bs U  + \nabla p
&= \bs 0,			\label{wave1ord_momc}	
\end{align}
\end{subequations}
which, after eliminating $\bs U$, can be written as a Helmholtz equation solely in the pressure,
%
\begin{equation}\label{Helmholtz}
-\Delta p - k_0^2 p = 0,
\end{equation}
where $k_0 = \omega/c$ is the isentropic wavenumber.

However,  the isentropic assumptions break down in the vicinity of solid walls modeled with the non-slip and isothermal boundary conditions
\begin{equation}\label{NSBC}
\bs U =\bs 0,
\qquad
T = 0,
\end{equation} 
which can be seen by the fact that it is only possible to set a vanishing normal component of the velocity, $\bs n\cdot\bs U=0$, as a boundary condition to system~\eqref{wave1ord}.
In the vicinity of a solid wall,  normal derivatives of the tangential velocity and the temperature will be large enough so that  not all visco--thermal terms can be ignored in equation~\eqref{linNS}. 
In order to account for the wall effects of viscosity and thermal conductivity, the isentropic system~\eqref{wave1ord} can in a narrow region close to the wall be replaced with a system of boundary-layer equations.

Assume that there is a flat solid wall located at the plane $y=0$, on which the non-slip and isothermal boundary conditions~\eqref{NSBC} should be imposed.
In a region close to the wall measured in terms of 
the boundary-layer thicknesses defined below,
the full system~\eqref{linNS} can be approximated, as derived in appendix~\ref{appBLeq}, with the acoustic boundary layer equations
\begin{subequations}\label{BLeqns}
\begin{align}
\i\omega\frac\rho{\rho_0} + \nablaT\cdot\bs u + \pder vy1 &= 0,				\label{BLeqnsMassC}
\\
\i\omega\bs u  + \frac1{\rho_0}\nablaT p 
- \nu \pder{\bs u}y2	&= \bs 0,								\label{BLeqnsMomCT}
\\
\pder py1 = 0,
\\
\i\omega\rho_0 c_V T + p_0\lp\nablaT\cdot\bs u + \pder vy1\rp 
- \kappa \pder Ty2 &= 0,
\label{BLeqnsEnerC}
\end{align}
\end{subequations}
where $\bs u = (u, 0, w)$ is the projection of the velocity vector $\bs U = (u, v, w)$ on the wall plane, and $\nablaT$ the corresponding projection of the operator  $\nabla$ , that is, 
\begin{equation}
\begin{aligned}
\nablaT = \lp\pder{}x1, 0, \pder{}z1\rp.
\end{aligned}
\end{equation}

In  appendix~\ref{appBLeq}, we construct exact solutions to system~\eqref{BLeqns} that satisfy the boundary conditions~\eqref{NSBC} and that exponentially, as $y\to+\infty$, approach the fields $\bs u^\infty(\bs r)$ and $p^\infty(\bs r)$, where $\bs r = (x, 0, z)$, which are assumed to be solutions to the isentropic equations~\eqref{wave1ord} evaluated at a point close to the wall but outside of the boundary layer.
These boundary-layer solutions can be written
\begin{subequations}\label{BLsol}
\begin{align}
\bs u &= \bs u^\infty(\bs r)\lp 1 - \e^{-(1+\i)y/\delta_V}\rp,				\label{BLusol}
\\
\frac{\rho_\text{e}}{\rho_0} &= \frac{\gamma-1}{\gamma}\frac{p^\infty(\bs r)}{p_0}\e^{-(1+\i)y/\delta_T},		\label{BLrhosol}
\end{align}
\end{subequations}
where 
\begin{equation}\label{excessdens}
\frac{\rho_\text{e}}{\rho_0} = \frac1{\rho_0}\lp\rho - \rho^\infty \rp= \frac1{\rho_0}\lp\rho - \frac{p^\infty}{c^2}\rp,
\end{equation}
is called the \emph{excess density}, and
\begin{equation}\label{dVdTdef}
\delta_V = \sqrt{\frac{2\nu}{\omega}},
\qquad
\delta_T =\sqrt{\frac{2\kappa}{\omega\rho_0 c_p}}
\end{equation}
the viscous and thermal \emph{acoustic boundary-layer thicknesses}.

Note that the boundary-layer solution as defined above, due to the small values of $\delta_V$ and $\delta_T$ quickly approaches the limit fields $\bs u^\infty$ and $p^\infty$, but that these values are not attained at any finite distance from the wall.
It would be possible to alter the approach and define a matched asymptotic expansion, but our purpose here is different; we will use the form of the boundary-layer solution in order to define an effective boundary condition to the isentropic equations to account for the boundary-layer effects.

\begin{remark}
The formation of boundary layers is due to the structure of the Navier--Stokes equations~\eqref{linNS} as a singularly-perturbed system.
That the thickness of the acoustic boundary layers in expressions~\eqref{dVdTdef} scales as the square root of the coefficients in the governing equations is a property that generally holds for  layers associated with singularly-perturbed equations.
For instance, the thickness of the classical  Prandtl-type of viscous boundary layer that develops over a flat plat subject to a steady free-stream flow parallel to the plate also scales as the square root of the viscosity.
However, in the Prandtl layer, the boundary-layer profile is not exponential like in expression~\eqref{BLusol}, and its thickness also grows as the square root of the distance from the leading edge~\cite{Schl87}.
\end{remark}

\section{Boundary layer effects modeled as a boundary condition}\label{BLasBC}

Our aim is to approximate the impact of the boundary layer with an effective wall boundary condition, which will be obtained by manipulations of  the mass conservation law.

The boundary-layer analysis described above provided formulas~\eqref{BLsol} for the tangential velocity and the excess density within the boundary.
Corresponding wall-normal velocity will be an order of magnitude smaller than the tangential velocity, as can be seen from  the scalings~\eqref{scalings} used to derive the boundary-layer equations.
Nevertheless, the wall-normal velocity $v$ at an arbitrary position $y = \tilde y$ within the boundary layer can be computed from the other variables by integrating equation~\eqref{BLeqnsMassC}, 
\begin{equation}\label{intBLMB}
\i\omega\int\limits_0^{\mathclap{\tilde y}}\frac\rho{\rho_0}\,\textit{dy} 
+ \int\limits_0^{\mathclap{\tilde y}} \nablaT\cdot\bs u\, \textit{dy} + v|_{y=\tilde y}  -  v|_{y=0} = 0
\end{equation}
where $v|_{y=0} = 0$ due to the non-slip boundary condition~\eqref{NSBC}, and where $\rho$ and $\bs u$ exponentially approach $\rho^\infty$ and $\bs u^\infty$ for increasing $\tilde y$.
Subtracting and adding $\rho^\infty/\rho_0$ (which is a function of wall position $\bs r$ only), using definition~\eqref{excessdens} and that  $v|_{y=0} = 0$,  we find that equation~\eqref{intBLMB} can be written
\begin{equation}\label{intBLMBb}
\i\omega\int\limits_0^{\mathclap{\tilde y}}\frac{\rho^e}{\rho_0}\,\textit{dy} 
+\i\omega\int\limits_0^{\mathclap{\tilde y}}\frac{\rho^\infty}{\rho_0}\,\textit{dy} 
+ \int\limits_0^{\mathclap{\tilde y}} \nablaT\cdot\bs u\, \textit{dy} + v|_{y=\tilde y} = 0
\end{equation}
An integration of the first  term in equation~\eqref{intBLMBb}, using formula~\eqref{BLrhosol},  yields
\begin{equation}\label{intrho}
\begin{aligned}
\i\omega\int\limits_0^{\mathclap{\tilde y}}\frac{\rho_\text{e}}{\rho_0}\,\textit{dy}
=
\delta_T\frac{\omega(\gamma-1)(1+\i)}{2\gamma p_0}p^\infty
\lp 1 - \e^{-(1+\i)\tilde y/\delta_T}\rp.
\end{aligned}
\end{equation}
Moreover, by expression~\eqref{BLusol}, we find that the third term in equation~\eqref{intBLMBb} can be evaluated as
\begin{equation}\label{intu}
\begin{aligned}
\int\limits_0^{\mathclap{\tilde y}} \nablaT\cdot\bs u\, \textit{dy} 
&= \int\limits_0^{\tilde y}\nablaT\cdot\bs u^\infty\lp1 - e^{-(1+i)y/\delta_V}\rp\,\textit{dy}
\\&
=\int\limits_0^{\tilde y}\nablaT\cdot\bs u^\infty\,\textit{dy} +
{\delta_V}\frac{1-\i}2 \nablaT\cdot\bs u^\infty\lp\e^{-(1+i)\tilde y/\delta_V }- 1\rp.  
\end{aligned}
\end{equation}
Substituting expressions~\eqref{intrho} and~\eqref{intu} into equation
~\eqref{intBLMBb}, we find that
\begin{equation}\label{intBLMBa}
\begin{aligned}
&\i\omega\int\limits_0^{\mathclap{\tilde y}}\frac{\rho^\infty}{\rho_0}\,\textit{dy} 
+ \int\limits_0^{\tilde y}\nablaT\cdot\bs u^\infty\,\textit{dy} 
+ v|_{y=\tilde y}
+ \delta_V\frac{\i-1}{2}\nablaT\cdot\bs u^\infty\lp 1 - \e^{-(1+\i)\tilde y/\delta_V}\rp
\\&\qquad\qquad
+ \delta_T\frac{\omega(\gamma-1)(1+\i)}{2\gamma p_0}p^\infty
\lp 1 - \e^{-(\i+1)\tilde y/\delta_T}\rp=0,
\end{aligned}
\end{equation}
which can be written
\begin{equation}\label{intBLMBc}
\begin{aligned}
&\i\omega\int\limits_0^{\mathclap{\tilde y}}\frac{\rho^\infty}{\rho_0}\,\textit{dy} 
+ \int\limits_0^{\tilde y}\nablaT\cdot\bs  u^\infty\,\textit{dy} 
+ v|_{y=\tilde y} + f(\tilde y) -  v_\text{W} = 0,
\end{aligned}
\end{equation}
where
\begin{subequations}\label{BLvs}
\begin{align}
f(y) &=
-\delta_V\frac{\i-1}{2}\nablaT\cdot\bs u^\infty \e^{-(1+\i) y/\delta_V}
- \delta_T\frac{\omega(\gamma-1)(1+\i)}{2\gamma p_0}p^\infty \e^{-(\i+1) y/\delta_T},
\label{BLvsy}\\
v_\text{W} = f(0)&=
- \delta_V\frac{\i-1}{2}\nablaT\cdot\bs u^\infty
- \delta_T\frac{\omega(\gamma-1)(1+\i)}{2\gamma p_0}p^\infty.
\label{BLvsW}
\end{align}
\end{subequations}

The function $f$ is of $O(\delta_V + \delta_T)$ at the wall and decays exponentially with its argument. 
%
We thus find that the integrated mass conservation law under boundary-layer approximation, equation~\eqref{intBLMB},  can be written as equation~\eqref{intBLMBc}, which, if the term $f(\tilde y)$ is ignored, all the effects of the boundary layer has been pushed into $v_\text{W}$. 
Recall that the solution to the isentropic system~\eqref{wave1ord} will be essentially constant in the normal direction close to a solid wall due to the lack of boundary layers.
Thus, equation~\eqref{intBLMBc} is essentially an integral form of the isentropic mass conservation law~\eqref{wave1ord_mc} (recall that $p= c^2\rho$ under isentropic assumptions) in which the wall normal velocity~\eqref{BLvsW}  is a perturbation of the non-penetration condition $\bs n\cdot\bs U = 0$ with coefficients of $O(\delta_V + \delta_T)$. 

The form of equation~\eqref{intBLMBc} and expression~\eqref{BLvsW} suggest that boundary layer effects could be taken into account by simply solving the isentropic system~\eqref{wave1ord} and replacing the normal isentropic wall boundary condition $\bs n\cdot \bs U = 0$ with
\begin{equation}
\bs n\cdot\bs U =  -v_W =  \delta_V\frac{\i-1}{2}\nablaT\cdot\bs U
+\delta_T\frac{\omega(\gamma-1)(1+\i)}{2\gamma p_0}p
\qquad\text{at $y=0$,}
\end{equation}
where we have used that $\nablaT\cdot\bs U = \nablaT\cdot \bs u$.
We thus propose the system 
%
\begin{subequations}\label{1ordBLBC}
\begin{align}
\frac{\i\omega }{c^2}p + \nabla\cdot\rho_0\bs U &= 0	&&\text{for $y>0$},		\label{1ordBLBC_mc}		
\\
\i\omega\rho_0 \bs U  + \nabla p &= 0	&&\text{for $y>0$},
\label{1ordBLBC_momc}
\\
\bs n\cdot\bs U &=  \delta_V\frac{\i-1}{2}\nablaT\cdot\bs U
+ \delta_T\frac{\omega(\gamma-1)(1+\i)}{2\gamma p_0}p &&\text{at $y=0$.}		\label{1ordBLBC_bc}
\end{align}
\end{subequations}
as a model for acoustic wave propagation over a wall at $y=0$ where thermal and viscous boundary layers form.

Instead of working with the first-order system~\eqref{1ordBLBC}, we suggest, for two reasons, to rewrite it as a second-order equation in $p$.
First, for numerical purposes, the number of unknowns will then be reduced to one scalar variable that can be treated with standard finite elements.
Second, the mathematical analysis of the boundary-value problem appears more straight-forward in the second-order formulation.

To work out the boundary condition  for the second order formulation that corresponds to condition~\eqref{1ordBLBC_bc}, we start by evaluating the wall normal component and the tangential divergence of equation~\eqref{wave1ord_momc} at the limit values for the boundary-layer,
\begin{subequations}\label{1ord_momc_split}
\begin{align}
\i\omega\rho_0 \bs n\cdot\bs U ^\infty+ \pder {p^\infty}n1 &= 0,
\\
\i\omega\rho_0 \nablaT\cdot\bs U^\infty + \DeltaT p^\infty &= 0.
\end{align}
\end{subequations}
Eliminating the velocity from system~\eqref{1ordBLBC}, where expressions~\eqref{1ord_momc_split} are used for the boundary condition~\eqref{1ordBLBC_bc}, we obtain the following second-order alternative to~\eqref{1ordBLBC},
\begin{subequations}\label{Hhsystem1}
\begin{align}
-k_0^2p - \Delta p &= 0
&&\text{for $y>0$,}
\\
-\delta_V\frac{\i-1}2 \DeltaT p + \delta_T k_0^2\frac{(\i-1)(\gamma-1)}2 p + \pder pn1 &= 0 \label{Hhsystem1_bc}
&&\text{at $y=0$.}
\end{align}
\end{subequations}

The solid-wall boundary conditions~\eqref{1ordBLBC_bc} or~\eqref{Hhsystem1_bc} are derived under the assumption of a flat wall, which makes the splitting in tangential and normal directions particularly easy.
However, such a splitting can also be carried out in the case of a smooth non-flat surface. 
In that case, the normal field vector $\bs n$ can be extended into the inside of the domain in the vicinity of the wall using the definition
\begin{equation}
\bs n (\bs x) = \frac{\nabla d(\bs x)}{\lvert\nabla d(\bs x)\rvert},
\end{equation}
where  $d(\bs x)$  is the \emph{wall distance function}~\cite{KrPa99}.
Taking derivatives in the directions of this extended normal field, we can in the vicinity of the wall split the gradient and divergence operators in their normal and tangential parts, analogously as in the case of a flat wall,
\begin{subequations}\label{tangopt}
\begin{align}
\nabla T &= \nablaT T + \bs n\pder Tn1,
\\
\nabla\cdot\bs U &= \nablaT\cdot\bs u + \bs n \cdot\pder{\bs U}n1,
\end{align}
\end{subequations}
where $\bs u = \bs U - (\bs U\cdot\bs n)\bs n $ and the tangential operators simply are defined through the expressions above.
However, the splitting of the Laplacian operator, needed in the derivation of the boundary-layer equations,  is more complicated in the curved-wall case,
\begin{equation}
\Delta T = \DeltaT T  + \pder Tn2  + H\pder Tn1 ,
\end{equation}
involving an extra term with $H = \nablaT\cdot\bs n$,  the sum of the principal curvatures of the level surface to $d$ that passes through the point of interest.
However, in many practical situations with a smooth wall, it is reasonable to assume that the minimal radius of the principal curvatures is much larger than the boundary layer thicknesses, which is of the order of 20--400~\textmu{m} in the audio range.
We will therefore here apply boundary conditions~\eqref{1ordBLBC_bc} or~\eqref{Hhsystem1_bc} also for nonplanar smooth boundaries, interpreting the tangential operators as in definitions~\eqref{tangopt}.
Due to the likely scale separation between boundary-layer thickness and the wall's radii of curvature, we conjecture that taking wall curvature into account in the boundary conditions would constitute in many cases at most a second-order correction to the boundary conditions derived above.

\section{Example: a cavity with  lossy walls}

Here we exemplify the use of wall boundary condition~\eqref{Hhsystem1_bc} in the context of an acoustic cavity problem.
Let the domain of the cavity, conceptually illustrated in figure~\ref{cavitypicture}, be denoted $\Omega\subset\RR^3$.
The cavity boundary $\partial\Omega$ consists of  a solid-wall part $\Gamma_\text{w}$ and a part $\Gamma_\text{io}$  where  sound waves can enter and exit.
We assume that both boundary parts are smooth and, as discussed in \S~\ref{BLasBC}, that the radii of the principal curvatures of the surface $\Gamma_{\text w}$ everywhere is much larger than $\delta_V$ and $\delta_T$.
Possible edges and corners of the cavity can therefore only be located at the interfaces between $\Gamma_\text{io}$ and $\Gamma_\text{w}$. 

The acoustic pressure amplitude in the cavity is then modeled by the boundary-value problem
\begin{subequations}\label{cavitysystem}
\begin{alignat}{2}
-k_0^2p - \Delta p &= 0
&\qquad&\text{in $\Omega$,}								\label{cavitysystem_eq}
\\
\i k_0 p + \pder pn1 &= 2\i k_0 g
&&\text{on $\Gamma_\text{io}$.}
\label{cavitysystem_io}
\\
-\delta_V\frac{\i-1}2 \DeltaT p + \delta_T k_0^2\frac{(\i-1)(\gamma-1)}2 p + \pder pn1 &= 0 \label{cavitysystem_w}
&&\text{on $\Gamma_\text{w}$,}
\\
\nT\cdot\nablaT p &= 0
&&\text{on $\partial\Gamma_\text{w}$,}	\label{cavitysystem_wbc}
\end{alignat}
\end{subequations}
Boundary condition~\eqref{cavitysystem_io} is a simple radiation (or impedance) condition, in which function $g$ supplies an incoming acoustic wave and where outgoing planar waves are absorbed.
Since wall boundary condition~\eqref{cavitysystem_w}  in itself constitutes a diffusion problem on the bounded surface $\Gamma_\text{w}$, an extra condition---``a boundary condition to the boundary condition''---is needed to close the system. 
Note that if $\Gamma_{\text w}$ would constitute the whole boundary, no such condition would be needed.
Here we choose perhaps the simplest alternative, the homogeneous Neumann condition~\eqref{cavitysystem_wbc}, where $\nT$ denotes the outward-directed unit normal on the boundary of the surface $\Gamma_{w}$.
(Note that $\nT$ is directed in the  \emph{tangent} direction of $\Gamma_\text{w}$).
Condition~\eqref{cavitysystem_wbc} will constitute a ``natural condition'' in the variational form and the power balance derived below; the interface $\partial\Gamma_{\text w}$ will be transparent in both expressions.
This case is the one treated in the well-posedness analysis of Appendix~\ref{wellposed}.
It would be mathematically possible instead to specify a  Robin condition like $\alpha p + \nT\cdot\nablaT p  = 0$.
In that case, an integral over $\partial\Gamma_{\text w}$ would appear in the variational form and the power balance, indicating a sink or source of power at $\partial\Gamma_{\text w}$. 
However, it is not clear if this condition makes physical sense, and there is no analysis to guide the choice of coefficient $\alpha$.
Assigning a Dirichlet condition at $\partial\Gamma_{\text w}$ is problematic, however, since it likely will lead to a jump discontinuity towards the $\Gamma_\text{io}$ side of the boundary and thus sharp gradients locally around $\partial\Gamma_{\text w}$.
The well-posedness theory in Appendix~\ref{wellposed} is not easily extended to the Dirichlet case, and it is not clear for us whether the case is  amenable to analysis at all.

Regardless of these mathematical issues, note that the boundary-layer approximations considered here breaks down at interfaces such as $\Gamma_{\text w}$, as well as at sharp corners within $\Gamma_\text{w}$, so ultimately, how to handle such interfaces is a modeling issue that is an interesting subject for further studies. 
\begin{figure}\centering
\includegraphics{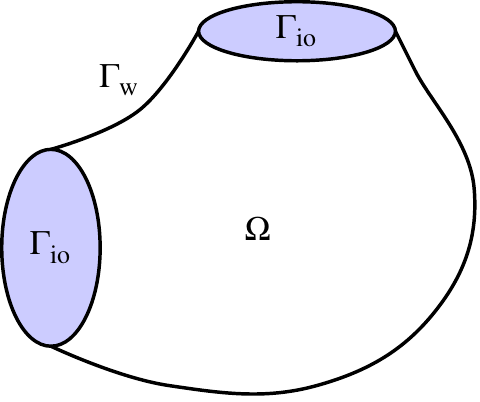}
\caption{Example cavity for problem~\eqref{cavitysystem}, viewed from the outside.
The interior is the domain $\Omega$.
Waves may enter and exit through the surfaces $\Gamma_\text{io}$, whereas the rest of the boundary $\Gamma_{\text w}$ consists of solid material, where thermal and viscous boundary layers form.
}\label{cavitypicture}
\end{figure}

Multiplying equation~\eqref{cavitysystem_eq} with a test function $q$, integrating by parts, and substituting boundary conditions~\eqref{cavitysystem_io}--\eqref{cavitysystem_wbc}, we find that solutions to the system~\eqref{cavitysystem} satisfies the variational expression
\begin{equation}\label{BLvarform}
\begin{aligned}
&-k_0^2\int_\Omega q p 
+ \int_\Omega\nabla q\cdot\nabla p
+ \i k_0 \int_{\Gamma_\text{io}} q p	
+ \delta_Tk_0^2\frac{(\i-1)(\gamma-1)}2\int_{\Gamma_\text{w}} q p
\\
&\qquad\qquad
+ \delta_V\frac{\i-1}2 \int_{\Gamma_\text{w}}\nablaT q\cdot\nablaT p
= 2\i k_0 \int_{\Gamma_\text{io}} q g.
\end{aligned}
\end{equation} 
\begin{remark}
In this section, as well as in Appendix~\ref{wellposed}, we do not explicitly specify measure symbols (such as $\textit{dV}$ or $\textit{dS}$, for instance) in the integrals, since the type of measure will be clear from the domain of integration.
\end{remark}

In Appendix~\ref{wellposed}, we define weak solutions  to system~\eqref{cavitysystem} using the variational form~\eqref{BLvarform} and show that the associated variational problem is well posed. 
Variational expression~\eqref{BLvarform} is also the starting point for finite-element discretizations, which can be carried out using the same standard elements as employed for the Helmholtz equation, that is, with  finite-element functions that are globally continuous and polynomials on each element of the triangulation.

Variational expression~\eqref{BLvarform} can also be used to derive a power balance law for system~\eqref{cavitysystem}, as follows.
Inserting the test function $q = \bar p$ (complex conjugate) in expression~\eqref{BLvarform}, we find that
\begin{equation}\label{e_VPp}
\int\limits_{\Omega}\lvert\nabla p\rvert^2 - k_0^2 \int\limits_\Omega \lvert p\rvert^2 
+ \i k_0\int\limits_{\mathclap{\Gamma_\text{io}}} \lvert p\rvert^2
+ \delta_Tk_0^2 \frac{(\i-1)(\gamma-1)}{2}\int\limits_{\Gamma_\text{w}} \lvert p\rvert^2
+ \delta_V\frac{\i-1}{2}\int\limits_{\Gamma_\text{w}}\lvert\nablaT p\rvert^2 = 2\i k_0 \int\limits_{\Gamma_\text{io}}\bar p g.
\end{equation}
The imaginary part of equation~\eqref{e_VPp} divided by $k_0$ is
\begin{equation}\label{e_impart}
\int\limits_{\mathclap{\Gamma_\text{io}}} \lvert p\rvert^2 
+ \delta_Tk_0\frac{\gamma-1}{2}\int\limits_{\Gamma_\text{w}} \lvert p\rvert^2
+ \frac{\delta_V}{2k_0}\int\limits_{\Gamma_\text{w}}\lvert\nablaT p\rvert^2
= 2\Re \int\limits_{\Gamma_\text{io}}\bar p g.
\end{equation}
Substituting identity
\begin{equation}
\lvert p - g\rvert^2 = \lvert p\rvert^2 + \lvert g \rvert^2 - 2\Re \bar p g,
\end{equation}
into equality~\eqref{e_impart} and dividing by $2\rho_0 c$, to obtain terms in units of power, we obtain 
\begin{equation}\label{e_powbal}
\frac1{2\rho_0c} \int\limits_{\Gamma_\text{io}} \lvert g\rvert^2
= \frac1{2\rho_0c}  \int\limits_{\Gamma_\text{io}} \lvert p - g\rvert^2
+ (\gamma-1)\frac{\delta_T\omega}{4\rho_0 c^2}\int\limits_{\Gamma_\text{w}}\lvert p\rvert^2
+ \frac{\delta_V}{4\omega\rho_0}  \int\limits_{\Gamma_\text{w}}\lvert\nablaT p\rvert^2.
\end{equation}
Expression~\eqref{e_powbal}  expresses that the incoming power equals the sum of the reflected power and the power losses due to the thermal and viscous boundary layers.

\section{Comparisons with classical results}

\subsection{Boundary-layer theory}\label{BLT}

Searby et al.~\cite{SeNiHaLa08} consider the problem of calculating the visco--thermal boundary-layer damping in acoustic cavities.
The authors review the boundary-layer theory and propose a two-step procedure in which a isentropic Helmholtz solver first calculates the pressure distribution on the solid surfaces.
In a second step, the total power loss from the visco-thermal boundary layers is calculated using boundary-layer theory.
The loss given by  Searby et al.~\cite[formulas~(6) and~(10)]{SeNiHaLa08}  agree with the two last terms in the power balance law~\eqref{e_powbal}.
Note, however,  that the procedure of Searby et al.\  does not predict, for instance, phase shift effects of the boundary layers, nor does it provide, as here, an explicit locally-reacting boundary condition, coupled to the interior problem.

Another approach is to consider a one-dimensional analysis, as presented by Rienstra \& Hirsch\-berg~\cite[\S~4.5]{RiHi15}, who analyze the thermal boundary layer in the case of a plane wave at normal incidence towards an isothermal wall and the viscous boundary layer in the case of a plane wave propagating parallel to the wall. 
The so-called \emph{displacement thicknesses} that their analysis yield are then applied to the case of a plane wave propagating in a wave guide. 
The final result is an expression for the complex wavenumber $k$ that governs the pressure amplitude, assumed to be of the form $p(z) = \e^{-\i k z}$, where the coordinate $z$ is along the axis of the wave guide, and in which the imaginary part of $k$ represents the visco--thermal damping. 
We will now show that if boundary-value problem~\eqref{cavitysystem} in its variational form~\eqref{BLvarform} is applied to such a case of a narrow waveguide, we obtain the same expression as Rienstra \& Hirsch\-berg. 

We consider the setup illustrated in figure~\ref{waveguidepicture}, where  $\Omega = S\times(0, \ell)$ is a cylindrical wave guide of length $\ell$.
The wave guide's cross section $S$ is fixed, of area $|S|$, and has a smooth boundary $\partial S$ of circumference $|\partial S|$. 
In order for the boundary-layer approximations to be valid, we assume that $\delta_V$ and $\delta_T$ are small compared to $\sqrt{|S|}$, and that $\sqrt{|S|}$ is of the same order as $|\partial S|$, precluding overly flattened geometries. 
Otherwise, the shape of $S$ can be arbitrary.
Let the waveguide be oriented along the $z$-axis and let $\Gamma_\text{io}$ be the two cross section planes located at $x=0$ and $x=\ell$, respectively.
Furthermore, let $g(x,0) = g_0$ and $g(x,\ell) = 0$ for a given number $g_0$.
According to boundary-layer theory and since the driving signal $g_0$ is constant at the inlet, we may make the ansatz that the pressure field is constant in each cross section, that is, that $p = p(z)$.
\begin{remark}
The assumption that $p = p(z)$ also constitutes a \emph{Galerkin approximation} of variational form~\eqref{BLvarform} such that $p$ and $q$ are constant over each cross section.
\end{remark}

Under this approximation, $\nabla p = \nablaT p = \bs e_z p'$, where $\bs e_z$ is a unit vector in the $z$ direction, and variational form~\eqref{BLvarform} reduces to 
\begin{equation}\label{BL1Dvarform}
\begin{aligned}
&\int_0^{\ell} q'\lb\lvert S\rvert + \delta _V\frac{\i-1}2\lvert\partial S\rvert\rb p'\,\textit{dz}
-k_0^2\int_0^{\ell} q\lb \lvert S\rvert - \delta_T\frac{\i - 1}2(\gamma-1) \lvert\partial S\rvert\rb p\, \textit{dz}
\\
&\qquad + \i k_0 \lvert S\rvert\bigl[ q(\ell)p(\ell) + q(0)p(0)\bigr] = 2\i k_0 q(0) g_0.
\end{aligned}
\end{equation}
Variational expression~\eqref{BL1Dvarform} holds for each test functions such that itself and its derivative is square integrable. 
In particular, for smooth test functions that vanish at $x = 0$ and $x = \ell$ (that is,  $q\in C_0^\infty(0,\ell)$), variational expression~\eqref{BL1Dvarform} reduces to
\begin{equation}
\int_0^{\ell} q'\lb\lvert S\rvert + \delta _V\frac{\i-1}2\lvert\partial S\rvert\rb p'\,\textit{dz}
-k_0^2\int_0^{\ell} q\lb \lvert S\rvert - \delta_T\frac{\i - 1}2(\gamma-1) \lvert\partial S\rvert\rb p\, \textit{dz} = 0.
\end{equation}
Integration by parts yield that 
\begin{equation}\label{wguideODE}
-\lb\lvert S\rvert + \delta _V\frac{\i-1}2\lvert\partial S\rvert\rb p'' -  k_0^2\lb\lvert S\rvert - \delta_T\frac{\i - 1}2(\gamma-1) \lvert\partial S\rvert\rb p = 0
\end{equation}
in $(0,\ell)$. 
The assumed boundary conditions imply that solutions of equation~\eqref{wguideODE} are of the form $p(z) = A e^{-i k z}$, which after substitution into equation~\eqref{wguideODE} implies that
\begin{equation}\label{k2BL}
k^2 = k_0^2\frac{2\lvert S\rvert - \delta_T(\i - 1)(\gamma-1) \lvert\partial S\rvert}{2\lvert S\rvert + \delta _V(\i-1)\lvert\partial S\rvert},
\end{equation} 
which is the same expression as obtained by Rienstra \& Hirsch\-berg~\cite[\S~4.5.3]{RiHi15}.
Note that expression~\eqref{k2BL} reveals that in order for $k$ to differ considerably from $k_0$, the surface-area to air-volume ratio $\lvert \partial S\rvert/\lvert S\rvert$ should be large.

\begin{figure}\centering
\includegraphics{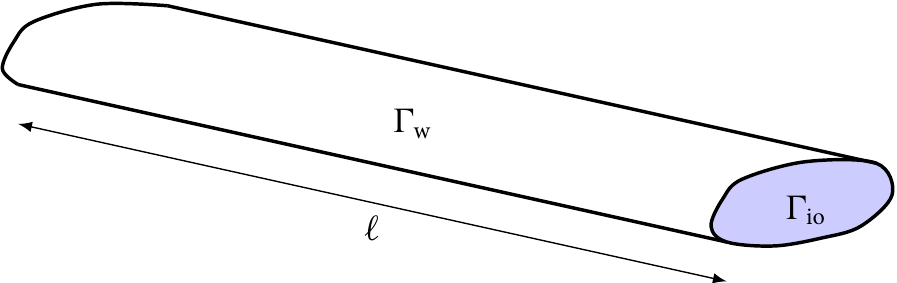}
\caption{Example of a cylindrical wave guide of the type considered in \S~\ref{BLT}}
\label{waveguidepicture}
\end{figure}

\subsection{Solutions for special geometries}

Instead of relying on boundary-layer analysis, it is possible to obtain exact or approximate solutions to the linearized Navier--Stokes equations for a few special geometries.
Keefe~\cite{keefe:84}, among others,  presents results based on Kirchhoff's classical solution for propagation of the first mode inside an isothermal cylinder with a circular cross section. 
The analysis is made in terms of the pressure and volume velocity over cross sections in the cylinder.
The wavelength is assumed to be sufficiently long for only the fundamental mode to propagate.
Note that, as opposed to what was the case in \S~\ref{BLT}, no assumption is made regarding the ratio between the tube radius and the boundary layer thickness.
For a circular tube of radius $a$, expression~(9) given by Keefe~\cite{keefe:84} implies the complex wavenumber 
\begin{equation}\label{eq:keefes_k2}
k^2 = k_0^2 \frac{1 + (\gamma-1)F_t}{1-F_v},
\end{equation}
where
\begin{equation}\label{FvFt}
\begin{aligned}
F_v &= \frac{2}{r_v\sqrt{-\i}} \frac{J_1(r_v\sqrt{-\i})}{J_0(r_v\sqrt{-\i})} ,\\
F_t &= \frac{2}{r_t\sqrt{-\i}} \frac{J_1(r_t\sqrt{-\i})}{J_0(r_t\sqrt{-\i})} ,
\end{aligned}
\end{equation}
in which $J_0$ and $J_1$ are the Bessel functions of order 0 and 1 and
\begin{align}
r_v &= a \sqrt{\omega/\nu}=
\sqrt{2}a/\delta_V,
\\
r_t &= r_v\sqrt{\nu \rho_0 c_p/\kappa} =
\sqrt{2} a/\delta_T.
\end{align}

We will now show that expression~\eqref{eq:keefes_k2} approaches the boundary layer expression~\eqref{k2BL} in the large radius or high frequency limit.
Such an analysis requires consideration of the behavior of the Bessel functions in expressions~\eqref{FvFt} in the limit of large arguments~\cite[Eq.~9.2.1]{AbSte72}. 
For the Bessel functions of a complex argument, asymptotic
expansions are meaningful only for a constant phase angle. 
For such a complex argument $w$, $-\pi<\arg(w)<0$,
\begin{equation}\label{J0expansion}
J_0(w) = \frac{1}{\sqrt{2\pi w}}e^{\i(w-\pi/4)} + \mathcal{O}(1/w)
\end{equation}
and
\begin{equation}\label{J1expansion}
J_1(w) = \frac{1}{\sqrt{2\pi w}}e^{\i(w-3\pi/4)} + \mathcal{O}(1/w).
\end{equation}
We can thus use the large-radius/high-frequency approximation
\begin{equation}
\frac{J_1(w)}{J_0(w)}\approx e^{-\i \pi/2} = -\i
\end{equation}
and find that
\begin{equation}
\begin{aligned}
F_v &\approx (1-i)\delta_V/a,\\
F_t &\approx (1-i)\delta_T/a,
\end{aligned}
\end{equation}
whence
\begin{equation}
\begin{split}
k^2 &= k_0^2 \frac{1 + (\gamma-1)F_t}{1-F_v} \approx
k_0^2\frac{1+(\gamma-1)(1-i)\delta_T/a}{1-(1-i)\delta_V/a} =\\
&= k_0^2\frac{2\pi a^2 + (\gamma-1)(1-i)2\pi a\delta_T}{2\pi a^2 - (1-i)2\pi a\delta_V} =
k_0^2\frac{2|S|+(\gamma-1)(1-i)|\partial S|\delta_T}{2|S|-(1-i)|\partial S|\delta_V},
\end{split}
\end{equation}
which means that expression~\eqref{eq:keefes_k2} yields in the large-radius/high-frequency limit the same expression~\eqref{k2BL} as when using the boundary-layer approximations. 

\begin{table}\centering
\caption{Air properties used for evaluation of damping models}\label{AirProp}
\begin{tabular}{lll}
\hline
Density 						  & $\rho_0$	&	$1.204$ $\text{kg}\cdot\textrm{m}^{-3}$	\\
Kinematic viscosity 		& $\nu$		 &    $1.506\cdot10^{-5}$ $\text{m}^2\cdot\text{s}^{-1}$  \\
Speed of sound 			&$c_0$		  &		$343.20$ $\text{m}\cdot\text{s}^{-1}$						\\
Prandtl number 				&$N_\text{Pr}$		&	$0.708$							\\
Specific heat, constant pressure				&  $c_P$					&$1.0054\times10^3$	J$\cdot\text{kg}^{-1}\cdot\text{K}^{-1}$	\\
Ration of specific heats			& $\gamma$					& $1.4$								\\
\hline
\end{tabular}
\bigskip

\includegraphics{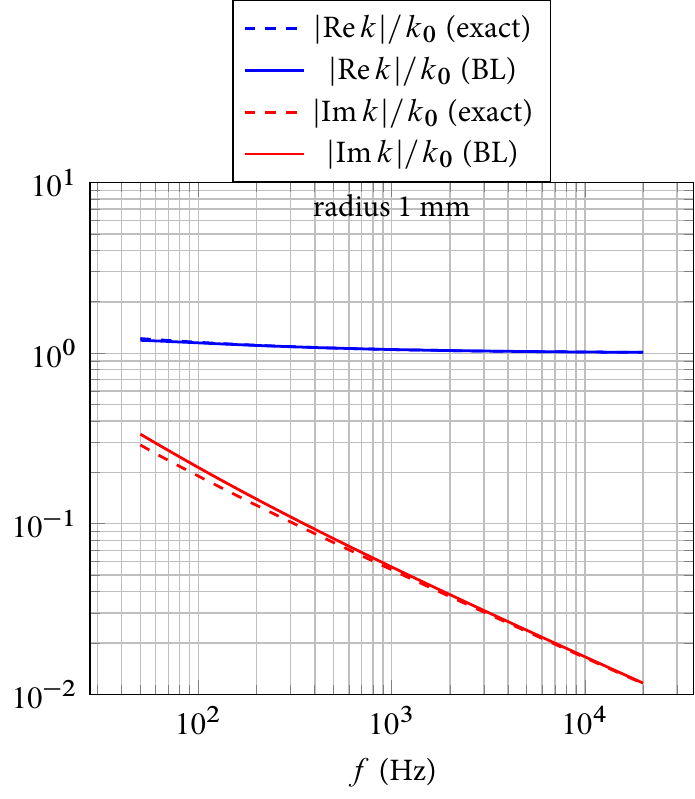}\includegraphics{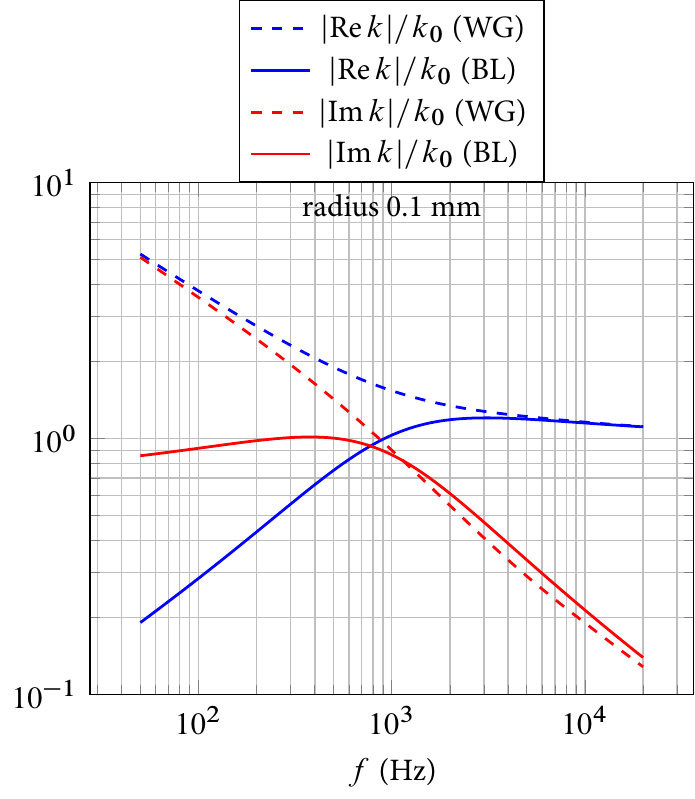}
\captionof{figure}{The real (blue) and imaginary (red) part of the relative wavenumber for lossy propagation in circular tubes of radii  1~mm (left) and 0.1~mm (right).
Dashed: ``exact'' wavenumber~\eqref{eq:keefes_k2} according to Keefe~\cite{keefe:84}.
Solid: boundary layer approximation~\eqref{k2BL}.
Air property parameters as in table~\ref{AirProp}.}\label{k_compare}

\end{table}

\subsection{Limits of applicability}

The dispersion relation~\eqref{k2BL} was derived from our approach of taking visco--thermal losses into account through boundary condition~\eqref{cavitysystem_w}.
Corresponding expression~\eqref{eq:keefes_k2} holds only for circular cross sections, but it holds also for tubes that are so narrow that the boundary-layer analysis ceases to be valid.
By comparing expressions~\eqref{k2BL} and~\eqref{eq:keefes_k2} for typical parameter values, it will therefore be possible to assess limits of applicability of our boundary-condition approach.

Figure~\ref{k_compare} shows two examples, for tubes of radii 1~mm and 0.1~mm, respectively, of applying the ``exact'' formula~\eqref{eq:keefes_k2} and the boundary-layer approximation~\eqref{k2BL}.
For the 1~mm tube, the agreement is almost perfect, whereas the plot for the 0.1~mm tube reveals that the asymptotics for low frequencies differ, so that the boundary-layer approximation will be inaccurate  for frequencies below, say, 2~kHz.
Thus, the boundary-layer approximations do break down, starting at low frequencies, for narrow enough tubes.
However, capillary tubes need to be very long and narrow in order for these inaccuracies to be noticeable. 

\section{Numerical case study: a generic compression driver}\label{sec:driver} 

This section attempts to demonstrate the usefulness of the above model in practical numerical computations. 
As have been alluded to in \S~\ref{BLT}, acoustic devices for which visco--thermal boundary-layer losses need to be taken into account are typically characterized by a large ratio of the area of solid surfaces to the air volume. 
One example of such a device is the \emph{compression driver}.
Due to the demands of high acoustic power, such drivers are commonly used in public address systems to feed mid-to-high-frequency horns.
In a compression driver, a stiff vibrating membrane is placed in a very shallow chamber, from which the sound exits  typically through a number of narrow slits in the radial or circumferential directions.
The quotient of the membrane area to the sum of the slits' area constitutes the \emph{compression ratio} of the driver.
A high compression ratio  improves the acoustic loading of the membrane, particularly at higher frequencies, and substantially increases the efficiency of the driver.  
A so-called \emph{phase plug} collects the acoustic output from the slits of the compression chamber and expands it to a circular waveguide, on which the throat of the horn will be mounted. 
The design of the slits and the phase plug is delicate in order to minimize the effects of internal resonances~\cite{Oclee-Brown2012}.
The presence of a shallow chamber and several thin slits
means that visco--thermal boundary-layer losses are potentially important to account for in a numerical simulation of a typical compression driver. 
However, as we will see, even a hybrid strategy, where the full Navier--Stokes equations are solved only in the narrow passages of the domain, whereas the pressure Helmholtz equation is used for the rest of the system, tends to lead to large problems and simulations which are expensive in terms of memory and CPU time.

Here, for the generic 3-inch compression driver design shown in figure~\ref{driver}, we compare such a hybrid strategy to a strategy where  the boundary losses are modeled by the proposed boundary condition \eqref{cavitysystem_w}.
The membrane diameter is 84~mm, the depth of the compression chamber is 0.5~mm, and the area of each of the 9 radial slits is 51~mm\textsuperscript2, which yields a compression ratio of 12.
The length of the phase plug is 25~mm and the diameter of the final wave guide is 38~mm.
The geometry of this driver is much simplified compared to  actual commercially available devices, but the dimensions above are representative for real-life drivers~\cite{Oclee-Brown2012}.
The air properties of table~\ref{AirProp} are used also here.

With the hybrid strategy, the Navier--Stokes equations~\eqref{linNS} are used in the compression chamber and the phase plug, and these equations are coupled to Helmholtz equation \eqref{cavitysystem_eq} in the waveguide.
The sound-hard walls are modeled with isothermal and no-slip boundary conditions in the Navier--Stokes region and with a homogeneous Neumann condition in the Helmholtz region.
Finite-element approximations of these equations are provided in the Acoustics Module (``Thermoviscous Acoustics, Frequency Domain'') of Comsol Multiphysics, which we use for the numerical experiments.
The use of a hybrid strategy is strongly recommended by the software provider to reduce the computationally cost, compared to using the Navier--Stokes equations everywhere~\cite{Co17AM}.
For the alternative strategy derived in this article, Helmholtz equation \eqref{cavitysystem_eq} is used throughout the whole domain, and with boundary conditions \eqref{cavitysystem_w}--\eqref{cavitysystem_wbc} on the walls of the compression chamber and phase plug. 
We implement these equations also in Comsol Multiphysics, using the software's so-called Weak Form PDE Interface, in which variational forms like expression~\eqref{varprob} directly can be specified~\cite[Ch.~16]{Co17Ref}.
A homogeneous Neumann bundary condition is used in the waveguide, as in the hybrid strategy.
A lossless model, with the homogeneous Neumann boundary condition on all walls, is used for comparison.
The membrane is regarded as a rigid piston oscillating with a fixed amplitude.
There are nine slits in the phase plug, and thus, due to its symmetry, it suffices to consider a $20^\circ$ segment with appropriate symmetry boundary conditions.

In all three models, the waveguide is terminated at a planar cross section $\Gamma_\text{out}$ supplied with boundary condition \eqref{cavitysystem_io} (with $g\equiv 0$) to model an infinite waveguide.
This condition, equivalent to imposing the acoustic impedance condition $Z_0 = \rho_0 c$, absorbs plane waves, and the acoustic power exiting the waveguide may then simply be computed as
\begin{equation}
P_\text{o} = \frac12\Re\int_{\Gamma_\text{out}} \bs n\cdot\bar{\bs u} p = \frac{1}{2\rho_0 c} \int_{\Gamma_\text{out}} |p|^2.
\label{eq:driver_Po}
\end{equation}
Note that non-planar modes does not propagate for the frequencies considered. 
At $10$~kHz, at the upper end of the frequency interval, the first non-planar mode decays with a factor of about $10^{-4}$ over the length of the waveguide, and therefore contributes a smaller error than the discretization.

\begin{figure}
\centering
\includegraphics[scale = 0.5]{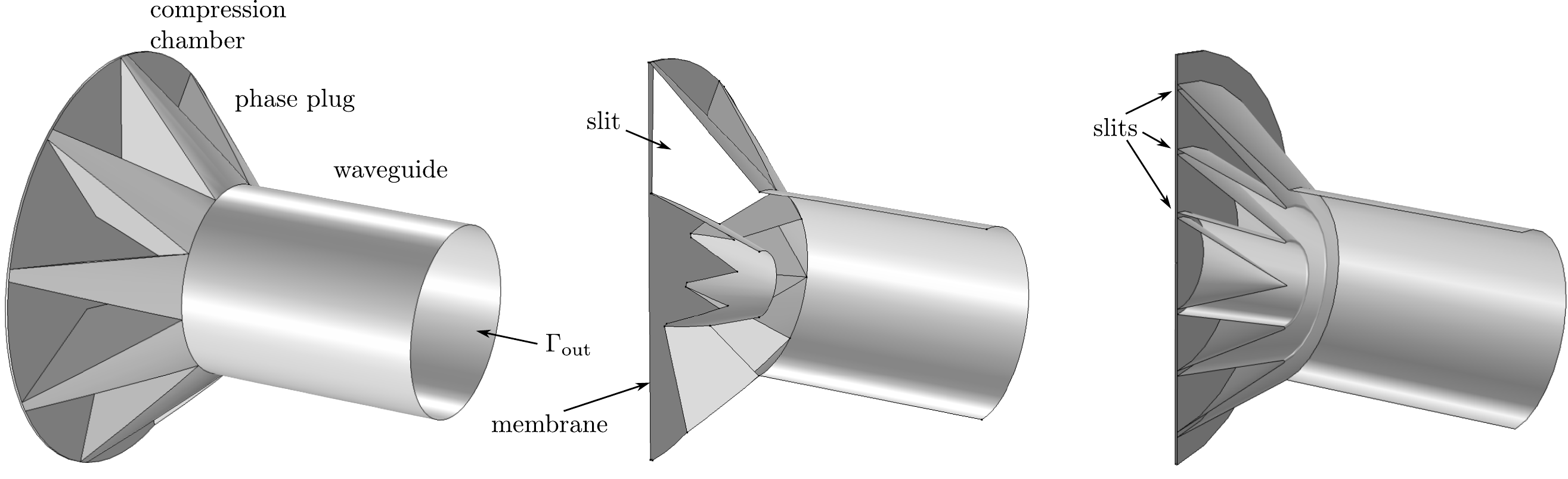}
\caption{
Left and middle: The generic compression driver geometry used to compare the  hybrid Navier--Stokes/Helmholtz strategy with our boundary-approximation approach. 
The circular membrane and the thin compression chamber is to the left, the phase plug in the middle consists of nine radial slits that expand into a circular wave guide.
Right: The rotationally symmetric compression driver geometry used to tune the mesh.
}\label{driver}
\centering
\includegraphics[scale = 1]{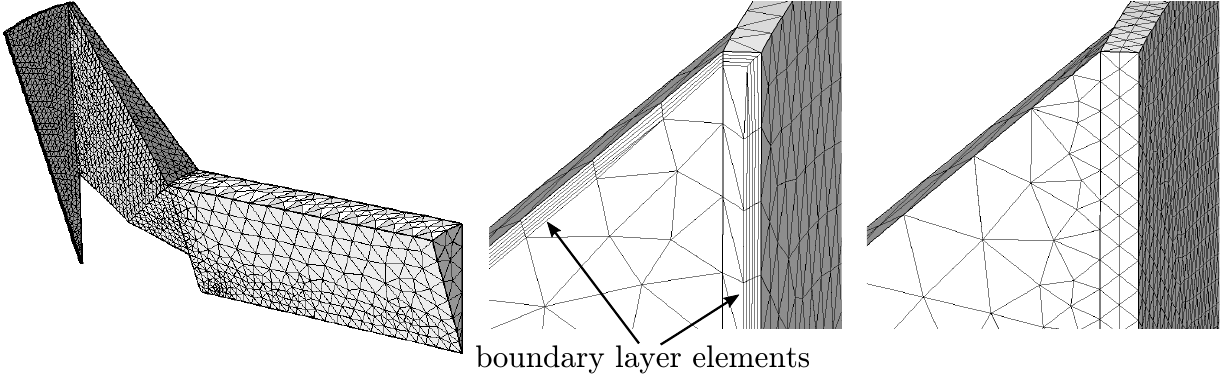}
\caption{Left: the mesh used in the hybrid strategy solution. Middle: close up of a cut through a part of the chamber and phase plug, where the boundary layer elements can be seen. Right: the same close up, but showing the mesh used in the boundary approximation solution. Note that this mesh, apart from the absence of boundary layer elements, is very similar to the mesh used for the hybrid strategy solution.
}\label{driver_mesh}
\end{figure}

%
%
\begin{figure}\centering
\includegraphics{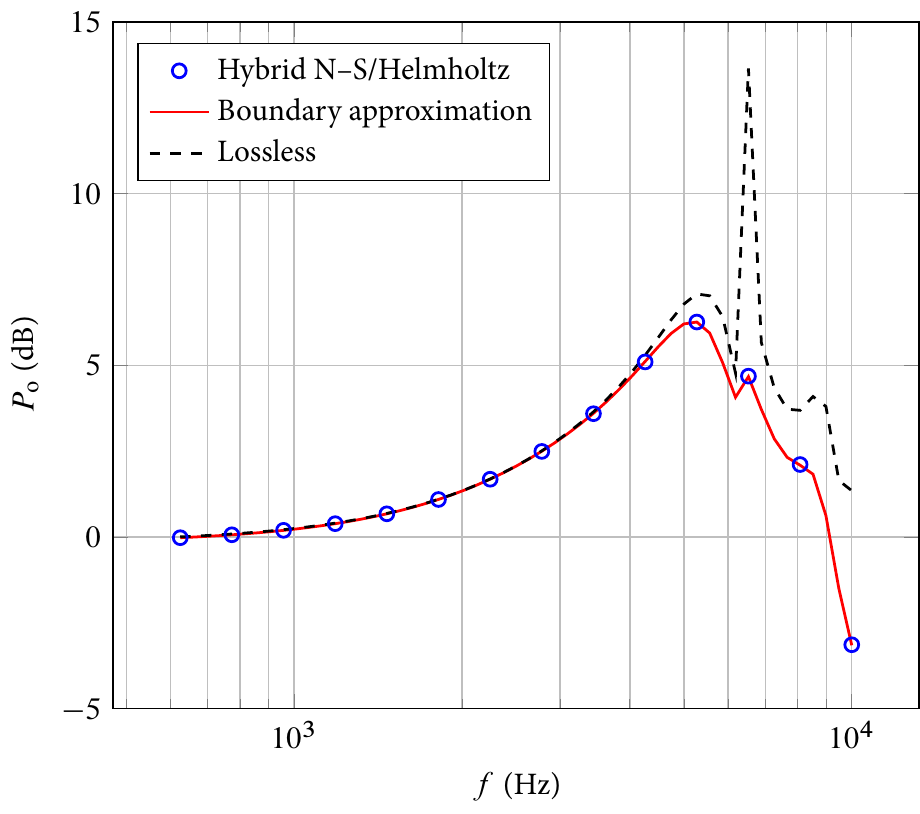}
\hfill
\caption{Output power $P_\text{o}$ for the compression driver in figure~\ref{driver}, computed with the Navier--Stokes/Helmholtz hybrid strategy, our boundary approximations, and when ignoring visco--thermal losses.
The 0-dB reference power is the output power at $f=625$~Hz.
}\label{Pout}

\bigskip
\captionof{table}{
Details from the Comsol log-files for the simulations of the compression driver in figure~\ref{driver}, on meshes tuned to give a relative accuracy of around $0.01$.
\\$^\ast$The solution time is the total wall clock time excluding I/O. 
Two computer cores were used for the boundary-approximation solution, while all 24 available cores were used for the hybrid solution.
}\label{table:driver}  
\bigskip
\begin{tabular}{l|rrr} 
& Degrees of freedom & Memory used & Solution time$^\ast$ \\ 
& & & per frequency \\ \hline
Hybrid strategy &  1 033 276 &  101 613 MB & 2 111  s  \\ 
Boundary approximation & 63 725 & 1 242 MB & 12 s \\ 
Quotient Hybrid/Boundary & 16.21 & 81.8 & 180 \\ 
\end{tabular}
\end{figure}

A good resolution of the rapidly varying velocity and temperature fields in the very thin boundary layers of the Navier--Stokes model requires a very fine mesh near the solid walls.
(The pressure, however, is essentially constant across the boundary layer.)
Anisotropic boundary-layer elements that are elongated in the tangential direction of the wall reduces the number of elements, compared to a uniform refinement, but the presence of the boundary layer will nevertheless lead to a large number of degrees of freedom.
We use quadratic elements for the acoustic pressure in all models, and cubic elements for the velocity components and the temperature in the solution of Navier--Stokes equations, relying on Comsol's choice of Taylor--Hood-like element orders.
To tune the mesh and estimate the accuracy of the solutions, we considered first a rotationally symmetric compression driver design in which the slits are annular rather than radial; see the right part of figure~\ref{driver}. 
In this reference design there are three slits, placed according to Smith's guidelines~\cite{Smith1953}, as further described in Ref.~\cite{Oclee-Brown2012}, and the compression ratio and other physical dimensions are identical to the design with the radial slits.
This geometry possesses axial symmetry and allows thus a 2D solution with an extremely fine mesh to be computed and used as a reference.
The chosen mesh settings for the 3D model are as follows. 
The maximum element length in interior of the model is $3.43 $~mm, which corresponds to one tenth of the wavelength at $10$~kHz. 
At the boundaries of the compression chamber and phase plug, there are six layers of elongated boundary elements, each with thickness $36.6$~\textmu m and maximum length $1.27 $~mm tangential to the boundary.
The boundary layer element thickness is chosen to resolve the viscous and thermal boundary layer effects, and should be compared to the viscous and thermal acoustic boundary-layer thicknesses $\delta_V$ and $\delta_T$.
These are $\delta_V = 21.9 $~\textmu m ($87.6$~\textmu m) and $\delta_T = 26.0$~\textmu m $(104$~\textmu m) at $10$~kHz ($625$~Hz), which is the upper (lower) end of the frequency interval considered.
Recall that the velocity and temperature, which are changing rapidly close to the boundaries, are discretized by polynomials of degree three.
The maximum relative difference in output power $P_\text{o}$ between the fine mesh 2D model and a $20^\circ$-segment 3D model on the chosen mesh is $\approx 0.01$.
The same 3D mesh settings are then used for the design to the left in figure~\ref{driver}, resulting in the mesh in figure~\ref{driver_mesh}, and it seems reasonable to expect an accuracy in the same order $\approx 0.01$ as for the axially-symmetric case.
A similar procedure is used to tune a coarser mesh, without elongated boundary layer elements, to use in our boundary approximation model.
Also in this case the mesh is chosen so that the estimated accuracy for the output power is $\approx 0.01$, to obtain a fair comparison between the hybrid and boundary approximation strategies.
In this model, the maximum element length is $0.40$~mm in the compression chamber, $1.27 $~mm at the boundaries of the phase plug, and $3.43$~mm in the interior of the phase plug and waveguide.

The results can be found in figure~\ref{Pout}.
First of all, we note that the output power is radically different when boundary effects are ignored, especially around the resonance at about 6.5~kHz. 
Moreover, the agreement between the hybrid strategy and our proposed boundary-approximation model is very good: the maximum relative difference in output power $P_\text{o}$ is less than $0.01$, and thus in the same order as the estimated relative numerical accuracy. 
Both the CPU time and memory consumption are radically lower for the boundary approximation model; the details can be found in table~\ref{table:driver}.

\section{Discussion}

Our proposed approach to account for visco--thermal boundary losses takes the form of the Wentzell condition~\eqref{Hhsystem1_bc} for the acoustic pressure.
The acoustical effects on both the amplitude and the phase of the thin viscous and thermal boundary layers,  which are in the order of 20--400~\textmu{m} in the audio range, are taken into account by this boundary condition. 
Thus, these layers do not need to be explicitly resolved by the mesh in numerical simulations, and the computational cost becomes essentially the same as for the lossless case. 
Therefore, we believe that this approach is a very attractive alternative to full or hybrid solutions to the linearized Navier--Stokes equations for many---maybe most---cases when visco--thermal boundary losses need to be taken into account.

The boundary-layer approximations break down for extremely thin capillaries or slits. 
However, as can be seen from figure~\ref{k_compare}, the break down starts at long wavelengths, which means that errors in our model will only be noticeable for devices with extreme geometries, such as very long submillimeter capillaries. 
We also saw in the numerical experiments of \S~\ref{sec:driver} that in the frequency range where the boundary-layer approximation becomes questionable, the losses were anyway very small.

The boundary-layer analysis presented here assumes that the radii of the principal curvatures of the wall is much larger than the boundary-layer thickness.
In particular, the effects of corners and edges along solid walls are not accounted for. 
Supported by the numerical experiments in \S~\ref{sec:driver}, it seems reasonable to assume that the effects of such geometric features may often only be of second order at low amplitudes.
The situation may be quite different, however, if the geometry includes a very large number of such features, or if significant wall roughness or microscale patterns in the order of the boundary-layer thickness are present.
The proposed boundary condition will likely not be a good model for such situations.
Out of scope for the present investigation is also nonlinear effects such as flow separation at edges, which will become significant at high amplitudes.

\clearpage
\appendix
\appendixpage

\section{ Boundary-layer equations}\label{appBLeq}

Here we outline the procedure to obtain the boundary layer approximation of system~\eqref{linNS} in the vicinity of a wall at $y=0$ on which boundary conditions~\eqref{NSBC} should be imposed.
We use the classical Prandtl approach of rescaling the equations and keeping only leading terms~\cite{Schl87}.
The analysis here is inspired by and constitutes a generalization of the 1D analysis presented by Rienstra \& Hirchberger~\cite[\S~4.5]{RiHi15}.
The first step is to split in the wall normal and tangential directions  all relevant quantities: the position vector; the velocity field; the nabla and the Laplacian operators; and the momentum balance equation~\eqref{linMomC}.
The purpose of this splitting is to prepare a rescaling of the variables in the different directions.

The splitting of the spatial points $\bs x$ and the velocity vector $\bs U$ in their tangential 
and wall normal parts 
will be denoted
\begin{equation}
\begin{aligned}
\text{$\bs x = (\bs r, y)$,\quad where\quad $\bs r = (x, z)$ is tangential and $y$ normal,} 
\\
\text{$\bs U= (\bs u, v)\,$,\quad where\quad $\bs u = (u, w)$ is tangential and $v$ normal.}
\\
\end{aligned}
\end{equation}
Likewise, corresponding splitting of the nabla and Laplacian operators will be written
\begin{equation}
\begin{aligned}
\nabla &= \lp\nablaT, \pder{}y1\rp, \quad\text{where}\quad   \nablaT = \lp\pder{}x1, \pder{}z1\rp,
\\
\Delta &= \DeltaT + \pder{}y2, \quad\text{where}\quad\DeltaT = \nablaT\cdot\nablaT.
\end{aligned}
\end{equation} 
With these definitions, system~\eqref{linNS} can now be split in normal and tangential parts with respect to the wall,
\begin{subequations}\label{splitlinNS}
\begin{align}
\i\omega\rho + \rho_0\lp\nablaT\cdot\bs u + \pder vy1\rp &= 0,				\label{splitlinMassC}
\\
\i\omega\bs u  + \frac1{\rho_0}\nablaT p 
- \nu\lp\DeltaT\bs u + \pder{\bs u}y2
+ \frac13\nablaT\lp\nablaT\cdot\bs u + \pder vy1\rp\rp 
&= \bs 0,								\label{splitlinMomCT}
\\
\i\omega v + \frac1{\rho_0}\pder py1
-  \nu\lp\DeltaT v +  \pder vy2 + \frac13\pder{}y1\lp\nablaT\cdot\bs u + \pder vy1\rp\rp 		&= 0,	\label{splitlinMomCN}
\\
\i\omega\rho_0 c_V T + p_0\lp\nablaT\cdot\bs u + \pder vy1\rp - \kappa\lp\DeltaT T + \pder Ty2\rp&= 0,
\label{splitlinEnerC}
\end{align}
\end{subequations}
Note that the single vector momentum equation~\eqref{linMomC} has been split to yield the two equations~\eqref{splitlinMomCT} and~\eqref{splitlinMomCN}.

Now we introduce the length scale $\tlength = c/\omega = 1/k_0$, where $k_0$ is the isentropic  wavenumber, in the tangential direction; the length scale $\delta$, to be determined later, in the wall-normal direction; and we nondimensionalize  coordinates and functions as follows:
\begin{equation}\label{scalings}
\begin{aligned}
\bs r' &= \frac{\bs r}\tlength,		
&\qquad
y ' &= \frac y\delta,
\\
\bs u'(\bs r', y') &= \frac1{\tlength\omega}{\bs u(\tlength \bs r', \delta y')},
&\qquad
v'(\bs r', y') &= \frac1{\delta\omega} v(\tlength \bs r', \delta y'),
\\
p'(\bs r', y') &= \frac1{\rho_0\tlength^2\omega^2} p(\tlength \bs r', \delta y'),
&\qquad
\rho'(\bs r', y') &= \frac1{\rho_0}\rho(\tlength \bs r', \delta y'),
\\
T'(\bs r', y') &= \frac{\rho_0 c_V}{p_o} T(\tlength \bs r', \delta y')).
\end{aligned}
\end{equation}
Nondimensionalization by rewriting equations~\eqref{splitlinMomCT} and~\eqref{splitlinMomCN} in these variables yields
\begin{subequations}\label{scalesplit}
\begin{align}
\i \rho'  +  \nablaT'\cdot\bs u' + \pder{v'}{y'}1  &= 0,							\label{scalesplitlinMassC}
\\
\i \bs u'  + \nablaT' p'  - \frac{\nu'}{\delta'^{2}}\lp\delta'^2\DeltaT' \bs u'  
+ \pder{u'}{y'}2 + \delta'^2 \frac13\nablaT'\lp\nablaT'\cdot\bs u' + \pder{v'}{y'}1\rp\rp &= \bs 0,
\label{scalesplitlinMomCT}
\\	
\delta'^2\i v' + \pder{p'}{y'}1 - \nu'\lp\delta'^2\DeltaT' v' 
+ \pder{v'}{y'}2 + \frac13\nablaT'\lp\nablaT'\cdot\bs u' + \pder{v'}{y'}1\rp\rp &= 0,
\label{scalesplitlinMomCN}
\\
\i T´ + \nablaT'\cdot\bs u' + \pder{v'}{y'}1 - \frac{\kappa'}{\delta'^2}\lp \delta'^2\DeltaT' T' + \pder{T'}{y'}2\rp &= 0,
\label{scalesplitlinEnerC}
\end{align}
\end{subequations}
where
\begin{equation}
\delta' = \frac\delta\tlength,
\qquad
\nu' = \frac{\nu}{\omega\tlength^2},
\qquad
\kappa' = \frac{\kappa}{\tlength^2\omega\rho_0 c_V}  = \frac{\gamma}{N_\text{Pr}}\nu'
\end{equation}
are the nondimensional vertical scaling, and the nondimensional viscosity and conductivity coefficients, respectively, and
where the last equality follows from the definition of the Prandtl number
\begin{equation}
N_\text{Pr} = c_P\rho_0 \frac{\nu}{\kappa} = \gamma \frac{\nu'}{\kappa'}.
\end{equation}

Substituting the small parameter $\epsilon=\nu' = \delta'^2$ into system~\eqref{scalesplit} (this will fix the vertical length scale) and ignoring $O(\epsilon)$ terms, the system reduces to
\begin{subequations}
\begin{align}
\i \rho'  +  \nablaT'\cdot\bs u' + \pder{v'}{y'}1  &= 0,							
\\
\i \bs u'  + \nablaT' p'  - \pder{u'}{y'}2  &= \bs 0,
\\	
\pder{p'}{y'}1  &= 0,
\\
\i T´ + \nablaT'\cdot\bs u' + \pder{v'}{y'}1 - \frac{\gamma}{N_\text{Pr}} \pder{T'}{y'}2&= 0,
\end{align}
\end{subequations}
which after transforming back to dimensional coordinates becomes
\begin{subequations}\label{BLEsystem}
\begin{align}
\i\omega\frac{\rho}{\rho_0} + \nablaT\cdot\bs u + \pder vy1 &= 0,			\label{BLEmc}
\\
\i\omega\bs u + \frac1{\rho_0}\nablaT p - \nu\pder{\bs u}y2 &= \bs 0,		\label{BLEumom}
\\
\pder py1 &= 0,																			\label{BLEvmom}
\\
\i\omega\rho_0 c_V T + p_0\lp\nablaT\cdot\bs u + \pder vy1\rp - \kappa\ \pder Ty2 &= 0.
\label{BLEenergy}
\end{align}
\end{subequations}
These are the acoustic boundary-layer equations, which, for small values of $\omega\nu $ (corresponding to small $\nu'$)  and for $N_\text{Pr} = O(1)$, approximate system~\eqref{splitlinNS}  in the vicinity of the wall;
that is, for wall-normal length scales of order $\sqrt{\nu/\omega}$ (corresponding to $\delta' = \sqrt{\nu'}$).
Recall that for air, as for many other gases, the Prandtl number is about unity.

Outside of a thin layer close to the wall,  the terms involving the viscosity and thermal conductivity parameter become small, and it becomes reasonable to replace system~\eqref{BLEsystem} with the isentropic  equations~\eqref{wave1ord}.  

System~\eqref{BLEsystem} can be solved exactly, subject to boundary conditions~\eqref{NSBC} and matching conditions for all variables as $y\to+\infty$ to an exterior isentropic solution to system~\eqref{wave1ord} evaluated at a  position $y = y_\delta$ close to the wall but outside the boundary layer.
We will denote these limiting values by $\bs u^\infty(\bs r)$, $p^\infty(\bs r)$, $\rho^\infty(\bs r)$ and so on.

Equation~\eqref{BLEvmom} simply implies that the pressure is constant in the vertical direction across the boundary layer,
\begin{equation}\label{pconst}
p(\bs r, y) = p^\infty(\bs r) 
\qquad\forall y>0.
\end{equation}
Moreover, the tangential components of equation~\eqref{wave1ord} evaluated at $y=y_\delta$ are
\begin{equation}\label{upinfeq}
\i\omega\rho_0 \bs u^\infty + \nablaT p^\infty = 0.
\end{equation}
Substituting equations~\eqref{pconst} and~\eqref{upinfeq} into equation~\eqref{BLEumom}, we can formulate the boundary-value problem 
%
\begin{subequations}
\begin{align}
\i\omega\bs u - \nu\pder{\bs u}y2 - \i\omega\bs u^\infty &= \bs 0,
\\
\bs u|_{y=0}  &= \bs 0,
\\
\lim_{y\to+\infty}\bs u &= \bs u^\infty,				
\end{align}
\end{subequations}
which has the solution
\begin{equation}
\bs u = \bs u^\infty(\bs r)\lp 1 - \e^{-(1+\i)y/\delta_V}\rp,
\end{equation}
where
\begin{equation}
\delta_V = \sqrt{\frac{2\nu}{\omega}},
\end{equation}
is the viscous boundary-layer thickness.

The remaining two equations~\eqref{BLEmc} and~\eqref{BLEenergy} combine into
%
\begin{equation}\label{MEcomb}
\i\omega\rho_0 c_V T -\i\omega\frac{p_0\rho}{\rho_0} - \kappa\pder Ty2 = 0.
\end{equation}
Dividing equation~\eqref{MEcomb} with $T_0$ and rewriting it in several steps, we find that
%
\begin{equation}\label{longone}
\begin{aligned}
0 &= \i\omega\lp\rho_0 c_V \frac T{T_0} - \frac{p_0\rho}{T_0\rho_0}\rp - \kappa\pder{}y2 \frac T{T_0}
= 	\i\omega\rho_0\lp c_V \frac T{T_0} - r\frac{\rho}{\rho_0}\rp - \kappa\pder{}y2  \frac T{T_0}
\\
&= \i\omega\rho_0c_V\lp  \frac T{T_0} - (\gamma-1)\frac{\rho}{\rho_0}\rp - \kappa\pder{}y2  \frac T{T_0}
= \i\omega\rho_0c_V\lp  \frac p{p_0} - \gamma\frac{\rho}{\rho_0}\rp - \kappa\pder{}y2  \lp \frac p{p_0} - \frac\rho{\rho_0}\rp
\\
&= -\lb
\i\omega\rho_0 c_p \lp\frac\rho{\rho_0} - \frac p{\gamma p_0}\rp 
- \kappa\pder{}y2 \lp\frac{\rho}{\rho_0} - \frac p{\gamma p_0}\rp
- \kappa\frac{1-\gamma}{\gamma}\pder{}y2 \frac p{p_0}
\rb
\\
&= -\lb
\i\omega\rho_0 c_p \lp\frac\rho{\rho_0} - \frac p{\gamma p_0}\rp 
- \kappa\pder{}y2 \lp\frac{\rho}{\rho_0} - \frac p{\gamma p_0}\rp
\rb,
\end{aligned}
\end{equation}
where ideal gas law~\eqref{idealgas} has been used in the second equality, expressions~\eqref{gamma} and \eqref{r} in the third equality, equation~\eqref{linIGL}  in the fourth equality,  definition~\eqref{gamma} in the fifth equality, and  equation~\eqref{BLEvmom} in the sixth equality.

Expression~\eqref{longone} is an equation in the \emph{relative excess density}
\begin{equation}\label{rhoedef}
\frac{\rho_\text{e} }{\rho_0}\defeq \frac{\rho}{\rho_0} - \frac p{\gamma p_0}
= \frac 1{\rho_0}\lp \rho - \frac p{c^2} \rp
= \frac1{\rho_0}\lp\rho - \frac{p^ \infty}{c^2} \rp
= \frac1{\rho_0}\lp\rho - \rho^\infty\rp,
\end{equation}
where expression~\eqref{c2} has been used in the second, property~\eqref{pconst} 
in the third, and the isentropic property~\eqref{c2ie} in the fourth equality.
The excess density is the difference between the actual density in the boundary layer and the limiting density outside the boundary layer.
The excess density thus vanishes when $y\to+\infty$, that is, when approaching the exterior isentropic conditions.
At $y=0$, it holds that
\begin{equation}
\rho_\text{e} = \frac{\rho}{\rho_0} - \frac p{\gamma p_0}
= \frac\rho{\rho_0} - \frac p{p_0} + \frac{\gamma-1}\gamma \frac p{p_0}
= - \frac T{T_0} + \frac{\gamma-1}\gamma \frac p{p_0}
= \frac{\gamma-1}\gamma \frac p{p_0},
\end{equation}
where the gas law~\eqref{linIGL} and boundary condition~\eqref{NSBC} are used in the third and fourth equality, respectively.

Altogether, the boundary-value problem for the excess density then becomes
\begin{subequations}\label{edBVP}
\begin{align}
\i\omega\rho_0 c_p \frac{\rho_\text{e}}{\rho_0} 
- \kappa\pder{}y2\frac{\rho_\text{e}}{\rho_0} &= 0
&&\text{for $y>0$,}
\\
\frac{\rho_\text{e}}{\rho_0}&\to 0		&&\text{as $y\to+\infty$,}
\\
\frac{\rho_\text{e}}{\rho_0}&= \frac{\gamma-1}\gamma \frac{p^\infty}{p_0}
&&\text{at $y=0$.}
\end{align}
\end{subequations}
The solution to problem~\eqref{edBVP} is
\begin{equation}
\frac{\rho_\text{e}}{\rho_0}= \frac{\gamma-1}{\gamma}\frac{p^\infty(\bs r)}{p_0}\e^{-(1+\i)y/\delta_T},
\end{equation}
where
\begin{equation}
\delta_T =\sqrt{\frac{2\kappa}{\omega\rho_0 c_p}}.
\end{equation}
is the thermal boundary-layer thickness.

\section{Well-posedness}\label{wellposed}

Here we establish well-posedness of weak solutions to boundary-value problem~\eqref{cavitysystem}, using variational expression~\eqref{BLvarform} as the starting point.
This section assumes familiarity with the Hilbert space approach to the analysis of elliptic partial differential equations~\cite{Wlok87}. 

We assume that domain $\Omega\subset\RR^3$ is  open, bounded, connected, and provided with a Lipschitz boundary.
Moreover, we assume that the closure of two nonempty, relatively open, smooth surfaces $\Gamma_\text{io}$ and $\Gamma_\text{w}$ comprise the boundary $\partial\Omega$.
The assumption of a Lipschitz boundary implies that possible nonsmooth portions of the boundary will be located at the interfaces between $\Gamma_\text{io}$ and $\Gamma_\text{w}$. 

The standard Sobolev space used to define weak solutions to the pressure Helmholtz equations is $H^1(\Omega)$,  the space of square integrable functions in which all partial derivatives are also square integrable.
However, in variational expression~\eqref{BLvarform}, there is also a tangential gradient operator in the boundary integral over $\Gamma_{\text w}$.
Therefore, we will also  require square integrability of the tangential derivatives over $\Gamma_{\text w}$. 
A natural norm on the space of functions considered for weak solutions to problem~\eqref{cavitysystem} is defined by 
\begin{equation}\label{Wdef}
\lVert\ p \rVert^2_W = k_0^2\int_\Omega \lvert p\rvert^2 
+ \int_\Omega \lvert\nabla p\rvert^2 
+ \delta_T(\gamma-1)k_0^2\int_{\Gamma_\text{w}}\lvert p\rvert^2 
+ \delta_V\int_{\Gamma_\text{w}} \lvert\nablaT p\rvert^2,
\end{equation}
in which the constants are chosen to make the various terms dimensionally consistent and to conveniently conform to the terms in variational form~\eqref{BLvarform}.
(Note that all terms are non negative since $\gamma>1$ for all gases.)
The  closure of the space of functions $\mathscr C^1(\overline\Omega)$ in this norm generates a strict subspace $W$ of $H^1(\Omega)$  in which the solutions will be defined.

The variational problem under consideration will then be the following:
\begin{equation}\label{varprob}
\begin{aligned}
&\text{Find $p\in W$ such that}
\\
&a(q, p) = \ell(q)\qquad\forall q\in W,
\end{aligned}
\end{equation}
where
\begin{subequations}\label{aldef}
\begin{align}
a(q,p) &= -k_0^2\int_\Omega q p 
+ \int_\Omega\nabla q\cdot\nabla p
+ \i k_0 \int_{\Gamma_\text{io}} q p		\label{aqpdef}
\\
&\qquad
+ \delta_Tk_0^2\frac{(\i-1)(\gamma-1)}2\int_{\Gamma_\text{w}} q p
+ \delta_V\frac{\i-1}2 \int_{\Gamma_\text{w}}\nablaT q\cdot\nablaT p,
\notag \\
\ell(q) &= 2\i k_0 \int_{\Gamma_\text{io}} q g.				\label{lqdef}
\end{align}
\end{subequations}

To analyze the properties of problem~\eqref{varprob}, we establish first a coercivity estimate and a uniqueness property in the following two lemmas.
\begin{lemma}[Coercivity]\label{Garding}
For any $p\in W$,
\begin{equation}
\big\lvert a(\bar p, p) + 2k_0\lVert p\rVert^2_{L^2(\Omega)}\big\rvert \geq \frac1{2\sqrt{13}}\lVert p\rVert_W^2
\end{equation}
\end{lemma}
\begin{proof}
Choose $q = \bar p$ (complex conjugate) in bilinear form~\eqref{aqpdef}, multiply with $(1-3\i/2)$, and take  the real part to obtain
\begin{equation}
\begin{aligned}
&\Re\lb\lp 1-\frac32\i \rp\lp a(\bar p, p) 
+ 2k_0\lVert p\rVert^2_{L^2(\Omega)} \rp\rb 
\\ &\qquad 
= \int\limits_{\Omega}\lvert\nabla p\rvert^2 + k_0^2 \int\limits_\Omega \lvert p\rvert^2 
+ \frac32 k_0\int\limits_{\mathclap{\Gamma_\text{io}}} \lvert p\rvert^2
+ \frac14\delta_T(\gamma-1)k_0^2\int\limits_{\Gamma_\text{w}}  \lvert p\rvert^2
+ \frac14\delta_V\int\limits_{\Gamma_\text{w}}\lvert\nablaT p\rvert^2. 
\\ &\qquad 
\geq \frac14 \lVert p\rVert_W^2, 
\end{aligned}
\end{equation}
which together with inequality
\begin{equation}
\begin{aligned}
& \Re\lb\lp 1-\frac32\i \rp\lp a(\bar p, p) 
+ 2k_0\lVert p\rVert^2_{L^2(\Omega)} \rp\rb&\leq \left\lvert1-\frac32\i \right\rvert
\left\lvert a(\bar p, p) 
+ 2k_0\lVert p\rVert^2_{L^2(\Omega)} \right\rvert
\\ 
&\qquad\qquad
=\frac{\sqrt{13}}2\left\lvert a(\bar p, p) 
+ 2k_0\lVert p\rVert^2_{L^2(\Omega)} \right\rvert
\end{aligned}
\end{equation}
yields the conclusion.
\end{proof}
\begin{lemma}[Injectivity]\label{unique}
For each $k_0> 0$, if $p\in W$ such that 
\begin{equation}\label{ahomeq}
a(q,p) = 0   \qquad\forall q\in W,
\end{equation}
then $p\equiv0$.
\end{lemma}
\begin{proof}
If $p\in W$ satisfies equation~\eqref{ahomeq}, then
\begin{equation}
\Im a(\bar p, p) = k_0\int_{\Gamma_\text{io}} \lvert p\rvert^2
+ \delta_T k_0^2\frac{\gamma-1}2 \int_{\Gamma_w} \lvert p\rvert^2 
+ \delta_V\frac12\int_{\Gamma_w} \lvert\nablaT p\rvert^2 = 0,
\end{equation}
which means that the trace of $p$ on $\partial\Omega$ vanishes, that is, $p\in H^1_0(\Omega)$.
Let  $\hat\Omega$ be an open extension of $\Omega$ in which $\Omega$ is compactly embedded; that is, $\hat\Omega\in\RR^3$ is  open, bounded, and connected such that $\overline\Omega\subset\hat\Omega$.
Since  $p\in H^1_0(\Omega)\subset W$ and satisfies equation~\eqref{ahomeq}, we may extend it by zero into a $p\in H^1(\hat\Omega)$ that, for each open set $K$ compactly embedded into $\hat\Omega$ ($\myol K\subset\hat\Omega$), satisfies
\begin{equation}
-k_0^2\int_\Omega qp + \int_\Omega\nabla q\cdot\nabla p = 0\qquad\forall q\in\mathscr C_0^\infty(K),
\end{equation}
which implies that 
\begin{equation}\label{Hhsat}
-k_0^2 p - \Delta p = 0\qquad\text{almost everywhere in $K$.}
\end{equation}
Since the extended $p$ satisfies property~\eqref{Hhsat} and vanishes identically in the open set $\hat{\Omega}\setminus\myol\Omega$, the unique continuation principle~\cite[Ch.~4.3]{Leis13} implies that $p\equiv0$ in $\Omega$.
\end{proof}

Lemma~\ref{Garding} and~\ref{unique} then imply that
\begin{theorem}
Problem~\eqref{varprob} has a unique solution for each $k_0>0$.
\end{theorem}
\begin{proof}
Continuity and coercivity of bilinear form~\eqref{aqpdef} follows from the Cauchy--Schwarz inequality and Lemma~\ref{Garding}, respectively.
The trace theorem on $H^1(\Omega)$~\cite[Th.~8.7]{Wlok87} yields that linear form~\eqref{lqdef} is bounded on $H^1(\Omega)$ and thus on $W\subset H^1(\Omega)$.
Moreover, the natural injection of $W$ into $L^2(\Omega)$ is compact due to compactness of the injection of $H^1(\Omega)$ into $L^2(\Omega)$.
The solution theory for general variational problems, as stated for instance by Wloka~\cite[Th.~17.11]{Wlok87}, then yields that that $a$ is \emph{Fredholm}; that is, either the homogeneous problem, equation~\eqref{varprob} with $\ell = 0$ has a nontrivial solution, or the inhomogeneous problem has a unique solution for each $\ell$.
However, Lemma~\ref{unique} says that the homogeneous problem only has the trivial solution for each $k_0>0$.
Thus, variational problem~\eqref{varprob} has a unique solution for each $k_0>0$.
\end{proof}

\section*{Acknowledgments}

Funding: This research was supported in part by the Swedish Research Council, Grant No. 621-2013-3706, the Swedish Foundation for Strategic Research, Grant No. AM13-0029, and  eSSENCE, a strategic collaborative eScience program funded by the Swedish Research Council.
The computations in \S~\ref{sec:driver} were performed on resources provided by the Swedish National Infrastructure for Computing (SNIC) at the High Performance Computing Center North (HPC2N).

\bibliographystyle{elsarticle-num}
\bibliography{mrabbrev,myabbrev,BLarticle}

\end{document}